\documentclass[10pt,twoside,a4paper]{amsart}
\usepackage{amsmath}
\usepackage{amsfonts}
\usepackage{amssymb}
\usepackage{amsthm}
\usepackage{newlfont}
\usepackage{graphicx}
\usepackage{amscd}
\usepackage{epsfig}

\theoremstyle{plain}
\newtheorem{Th}{Theorem}%[section]

\newtheorem{Lem}[Th]{Lemma}
\newtheorem{Prop}[Th]{Proposition}
\newtheorem*{GIS}{Geometric Integrability Scheme}

\newtheorem*{4Dconsistency}{Four Dimensional Consistency of the Geometric
Integrability Scheme}
\theoremstyle{definition}
\newtheorem{Def}{Definition}%[section]
\theoremstyle{remark}
\newtheorem*{Rem}{Remark}%[section]
\numberwithin{equation}{section}
\newcommand{\PP}{{\mathbb P}}

\newcommand{\DD}{{\mathbb D}}
\newcommand{\ZZ}{{\mathbb Z}}

\newcommand{\AAf}{{\mathbb A}}
\newcommand{\bx}{\boldsymbol{x}}

\newcommand{\bX}{\boldsymbol{X}}
\newcommand{\tbX}{\tilde{\bX}}

\newcommand{\tH}{\tilde{H}}
\newcommand{\tQ}{\tilde{Q}}

\newcommand{\D}{{\Delta}}

\begin{document}

\title[Geometric algebra and  quadrilateral lattices]
{Geometric algebra and  quadrilateral lattices}

\author{Adam Doliwa}

\address{Adam Doliwa, Wydzia{\l} Matematyki i Informatyki,
Uniwersytet Warmi\'{n}sko-Mazurski w Olsztynie,
ul.~\.{Z}o{\l}nierska~14, 10-561 Olsztyn, Poland}

\email{doliwa@matman.uwm.edu.pl}

\date{\today}
%\subjclass[2000]{37K10, 37K20, 37K25, 37K35, 37K60, 39A10}

\begin{abstract}
Motivated by the fundamental results of the geometric algebra we 
study quadrila\-teral lattices in projective spaces over division rings. After
giving the noncommutative discrete Darboux equations we discuss differences and
similarities with the commutative case. Then we consider the 
fundamental
transformation of such lattices in the vectorial setting and we show the
corresponding permutability theorems. We discuss also the
possibility of obtaining in a similar spirit a
noncommutative version of the B-(Moutard) quadrilateral lattices. 
  
\end{abstract} 
%{\it 2001 PACS:} 02.30.Ik, 02.40.Dr, 05.45.Yv, 04.60.Nc, 02.40.Hw,
\maketitle

%\tableofcontents

\section{Introduction}

\subsection{Integrable discrete geometry}
In the course of last ten years many results of the classical
geometric approach to 
integrable partial differential equations \cite{Sym,RogersSchief}
has been transfered to the discrete
setting (see \cite{DS-EMP} and references therein). The
key role in the theory has been attributed to the
multidimensional quadrilateral lattice \cite{MQL}, which is the discrete analog
\cite{Sauer} of a conjugate net \cite{Darboux-OS,Eisenhart-TS}. 
It turns out that integrability of the quadrilateral lattice is encoded in a
very simple geometric statement, visualized on Figure~\ref{fig:TiTjTkx}. 
\begin{GIS}%[The geometric integrability scheme] \label{lem:gen-hex}
Consider points $x_0$, $x_1$, $x_2$ and $x_3$ in general position in the
projective space $\PP^M$, $M\geq 3$. On
the plane $\langle x_0, x_i, x_j \rangle$, $1\leq i < j \leq 3$ choose a point
$x_{ij}$ not on the lines  $\langle x_0, x_i \rangle$, $\langle x_0,x_j
\rangle$ and $\langle x_i, x_j \rangle$. Then there exists the
unique point $x_{123}$
which belongs simultaneously to the three planes 
$\langle x_3, x_{13}, x_{23} \rangle$,
$\langle x_2, x_{12}, x_{23} \rangle$ and
$\langle x_1, x_{12}, x_{13} \rangle$.
\end{GIS}
\begin{figure}[h!]
\begin{center}
\includegraphics{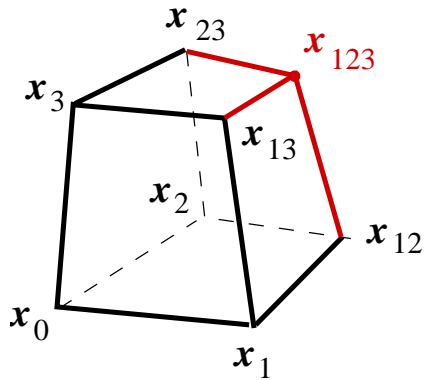}
\end{center}
\caption{The Geometric Integrability Scheme}
\label{fig:TiTjTkx}
\end{figure}
Despite of (or, thanks to) extremely simple formulation of
the Geometric Integrability Scheme, 
the corresponding nonlinear discrete system 
(the discrete Darboux equations) turns out to be the generic 
discrete system integrable by the nonlocal $\bar\partial$-dressing
method \cite{BoKo}.  Also the finite-gap integration scheme, a standard tool in
the integrable systems theory \cite{BBEIM}, can be applied to
that system in its pure form \cite{AKV}. 
We mention that the \emph{differential} Darboux equations, which have appeared 
first in projective
differential geometry of multidimensional conjugate nets \cite{Darboux-OS}, 
play an important role \cite{BoKo-N-KP,DMMMS} in the multicomponent
Kadomtsev--Petviashvilii (KP) hierarchy, 
which is commonly considered \cite{DKJM,KvL}
as the fundamental system of equations in integrability theory.

Integrable reductions of the quadrilateral lattice (and thus of the discrete
Darboux equations) arise from additional constraints which are compatible with
the geometric integrability scheme. In \cite{q-red,BQL,CQL} we isolated the
incidence geometry theorems which are responsible for the basic
reductions of the quadrilateral lattice:  
B- and C-reductions providing geometric interpretation for BKP and CKP 
hierarchies \cite{DJKM-CKP}, and the so called quadratic reduction. 

On the
geometric level there is no essential difference between the quadrilateral
lattice construction and between its Darboux-type transformations
\cite{MDS,KoSchief2,TQL,MM}. In particular, all classical transformations of
conjugate nets \cite{Eisenhart-TS,DMMMS} have found their quadrilateral lattice
analogs and have been shown to be reductions of the discrete analog of the 
fundamental transformation of Jonas.

Although the geometric integrability scheme was initially considered for real
projective spaces, it is valid to projective spaces over other
fields. In particular, finite field version 
together with the algebro-geometric method of construction of solutions to the 
corresponding discrete equations has been
given in \cite{DBK-JPA,BD-CMP}. 
The main idea of the present paper is that the geometric integrability scheme
remains valid in projective spaces over division rings (called also skew fields,
for details see \cite{Cohn}), whose simplest example are
quaternions. We would like to mention that division rings appear naturally in
a generalization of the notion of determinant to matrices with noncommutative
entries \cite{GGRW-quasideterminants}. The so called quasideterminants have been
effective in many areas including noncommutative symmetric functions,
noncommutative integrable systems, quantum algebras and Yangians,
noncommutative algebraic geometry. Last but not least, the division ring of
formal pseudodifferential operators lies in the heart of the Sato approach
\cite{Sato} to integrable systems, see also \cite{Parshin}. We should warn the
Reader that a ring of square matrices usually is not a division ring (the sum of
two invertible matrices does not have to be invertible or the zero matrix).
Also, by the Wedderburn theorem, finite division rings are commutative.

The subject of noncommutative versions of integrable systems was studied in
the literature in many papers, see, for example \cite{DM-H,LiNimmo} and references
therein; we would like to stress that in the present paper
noncommutativity is considered only on the level of dependent variables, i.e., 
the independent variables are still commutative ones.
 In relation to our work we would like to mention the paper
\cite{Nimmo-NCKP} where the noncommutative discrete KP equation was considered,
and the papers \cite{BobSur-NC,Schief-LPR}. Moreover, in \cite{Sergeev} 
a quantization of the discrete Darboux equations was investigated. It should be
also mentioned that already in the paper \cite{BoKo} the discrete Darboux
equations, together with some of their transformations,
were considered in the matrix version within the non-local
$\bar\partial$-dressing method, thus in the
noncommutative setting (for the differential matrix Darboux-Manakov--Zakharov
equations see \cite{ZakhMan}).

The paper is constructed as follows. In Section \ref{sec:QL} we study the
multidimensional quadrilateral lattices in projective spaces over division
rings. In particular, we consider the corresponding discrete noncommutative
Darboux equations together with the corresponding linear problem, and 
we discuss differences and similarities with the commutative case. 
In section \ref{sec:FT} we
give the vectorial fundamental transformation for such lattices. Finally, 
in Section \ref{sec:BC} we study
possibility of the geometric
generalization of the B-quadrilateral lattices (and thus of the
discrete BKP equations) to the noncommutative setting. 
We show that the additional incidence geometry structures which imply
integrability of the B-quadrilateral lattices \cite{BQL}
force the division ring to be commutative. 

The main results of the paper were presented in my talk
\emph{Geometric algebra and quadrilateral lattices} during the ISLAND 3 
(Integrable Systems: Linear and Nonlinear Dynamics) conference
\emph{Algebraic Aspects of Integrable Systems}, Port Ellen, 
Isle of Islay, Scotland (July, 2007).

\subsection{Some basic facts from geometric algebra}
Because the intended target
of the paper consists of 
specialists from integrable systems theory we start from
presenting some basic facts on the interplay between
incidence geometry axioms and the corresponding algebraic structures (for
details see \cite{Hilbert,Artin,Baer,BeutRos,BeukenhoutCameron-H}).

A \emph{projective plane} is a set, whose elements are called
\emph{points} and a set of
subsets, called \emph{lines}, satisfying the following four axioms:
\begin{itemize}
\item[P1] Two distinct points lie on one and exactly one line;
\item[P2] Two distinct lines meet in precisely one point;
\item[P3] There exist four points with no three collinear.
\end{itemize}
\begin{figure}[h!]
\begin{center}
\includegraphics[width=7cm]{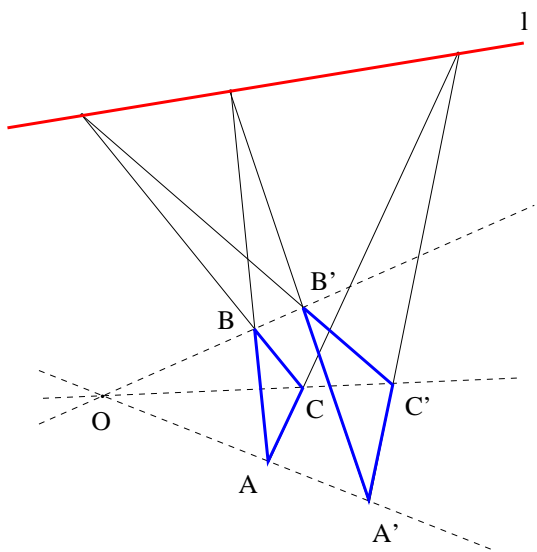}
\end{center}
\caption{The Desargues configuration: the triangles $\triangle ABC$ and
$\triangle A^\prime B^\prime C^\prime $ are perspective from the point $O$,
and are perspective from the line $l$.}
\label{fig:Desargues}
\end{figure}
It is known that axioms P1-P3 make possible to introduce on the plane
coordinates from an
algebraic structure called the \emph{ternary ring}. 
If, in addition to P1-P3, \emph{the Desargues axiom} holds: 
(see Figure~\ref{fig:Desargues}):
\begin{itemize}
\item[P4]
If two triangles are perspective from a point then they are
perspective from a line;
\end{itemize}
then axioms P1-P4 imply possibility of coordinatization of the plane
in terms of a 
division ring.
If, instead, one adds to the axioms P1-P3 the so called \emph{Pappus' axiom}:
\begin{itemize}
\item[P4']
If the six vertices of a hexagon lie alternately on two lines, then the three
points of intersection of pairs 
of opposite sides are collinear;
\end{itemize}
\begin{figure}
\begin{center}
\includegraphics[width=8cm]{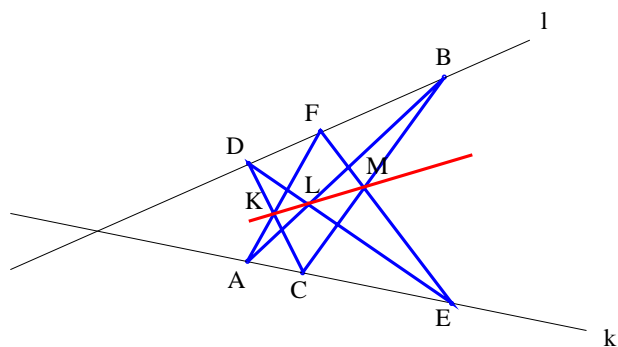}
\end{center}
\caption{The Pappus configuration: the vertices of the hexagon $ABCDEF$
lie alternately on two (coplanar) lines $k$ and $l$, and the three
points $K$, $L$, $M$ of intersection of pairs 
of opposite sides are collinear}
\label{fig:Pappus}
\end{figure}
then one has coordinates in commutative division ring, i.e. in a field. 

For more dimensional projective spaces the basic incidence axioms, analogous to
P1--P3, are enough to show that the spaces are actually coordinatized by division
rings, i.e., there is no need for the Desargues axiom (which becomes a theorem).
In order to have a projective geometry over a field one has to add the Pappus
axiom (or its equivalent formulations). 

\section{Quadrilateral lattice in spaces over division
rings (affine description)}
\label{sec:QL} 

Because the Geometric Integrability Scheme is valid in projective spaces over
division rings,  this motivates us to consider quadrilateral lattices in such
spaces.

\subsection{The Laplace and Darboux equations}
Consider a multidimensional quadrilateral lattice, i. e., a mapping
$x :\ZZ^N \rightarrow \PP^M(\DD)$
with all the elementary quadrilaterals planar \cite{MQL}; here 
$\ZZ^N$ is  $N\geq 3$ dimensional integer lattice, 
and $\PP^M(\DD)$ is $M\geq N$ dimensional right projective space over 
division ring $\DD$ (we multiply vectors by scalars from right). 
It turns out that the theory of quadrilateral lattices in
spaces over division rings does not differ considerably from the standard case
where $\DD$ was assumed to be commutative. One should be only careful with the
order of coefficients.

Below we will use the affine description of the quadrilateral lattice. Recall
that the affine space $\AAf^M = \PP^M \setminus H_\infty$ is the projective
space with removed a fixed hyperplane 
$H_\infty\subset\PP^M$ (called the hyperplane at infinity; see, for example
\cite{Coxeter-PG}). Two lines of $\AAf^M$
called parallel if they intersect in a point of $H_\infty$.

In the affine gauge the 
lattice is represented by a mapping
$\bx: \ZZ^N \rightarrow \DD^M$, the planarity
condition can be formulated in terms of the Laplace equations
\begin{equation}  \label{eq:Laplace}
T_i T_j\bx - \bx =(T_i\bx - \bx) A_{ij}+
(T_j\bx - \bx)A_{ji},\quad i\not= j, \qquad i,j=1 ,\dots, N,
\end{equation}
where $T_i$ is the translation operator in the $i$-th direction. 
Then the coefficients $A_{ij}:\ZZ^N\to\DD$ satisfy, by compatibility of 
the system \eqref{eq:Laplace}, 
\begin{equation} \label{eq:MQL-A}
A_{jk} T_k A_{ji} = 1 + (A_{ji} - 1)T_j A_{ik} + (A_{jk} - 1)T_j A_{ki},
\qquad i, j, k \quad \text{distinct}.
\end{equation}
The $i\leftrightarrow k$ symmetry of RHS of \eqref{eq:MQL-A} implies existence
of the potentials $H_{i}$, $i=1, \dots , N$, 
(called the Lam\'e coefficients) such that
\begin{equation}   \label{def:A-H}
A_{ij}= T_i\left( H_i^{-1} T_j H_i\right) , \quad i\ne j .
\end{equation}

If we introduce the suitably scaled tangent 
vectors $\bX_i:\ZZ^N\to\DD^M$, $i=1,...,N$, by equations
\begin{equation}  \label{def:HX}
\D_i\bx = \bX_i T_iH_i,
\end{equation}
(here $D_i = T_i - \mathrm{id}$)
and the rotation coefficients $Q_{ij}:\ZZ^N\to\DD$, $i\ne j$, by
\begin{equation} \label{eq:lin-H}
\D_iH_j = Q_{ij} T_iH_i, \quad i\ne j ,
\end{equation}
then equations \eqref{eq:Laplace} can be rewritten as a
first order system
\begin{equation} \label{eq:lin-X}
\D_j\bX_i = \bX_j T_j Q_{ij},    \quad i\ne j .
\end{equation}
%\begin{figure}
%\begin{center}
%\includegraphics[width=5cm]{forward.eps} %\hskip 1cm
%\includegraphics[width=8cm]{back.eps}
%\caption{Definition of the forward and backward data}
%\end{center}
%\label{fig:forward-back}
%\end{figure}
The compatibility condition for the system~\eqref{eq:lin-X} (or its
adjoint~\eqref{eq:lin-H})
gives the following form of the MQL (or discrete Darboux) equations
\begin{equation} \label{eq:MQL-Q}
\D_k Q_{ij} = Q_{kj} T_k Q_{ik},    \qquad i, j, k \quad \text{distinct}.
\end{equation}
\begin{Rem}
The above equations (up to small modification which, in our language, results
from considering left-vector spaces)
appeared first in the matrix setting in \cite{BoKo}.
\end{Rem}

An important geometric fact, which lies in the heart of integrability of
the quadrilateral lattice, is the multidimensional consistency of the
geometric integrability scheme. Its four dimensional version reads as follows
(see Fig.~\ref{fig:4Dconsistency}).
\begin{figure}
\begin{center}
\includegraphics[width=8cm]{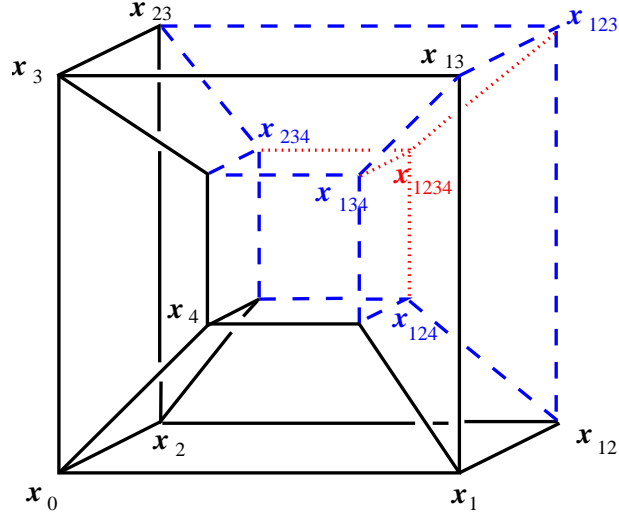}
\end{center}
\caption{The four dimensional consistency of the geometric integrability scheme:
Starting from the initial quadrilaterals (solid lines) in the fist step of the
construction (dashed lines) one obtains four hexahedra shearing the vertex
$x_0$. The second step of the construction (dotted lines) gives another
four hexahedra
exhausting this way all the hexahedra of the hypercube.}
\label{fig:4Dconsistency}
\end{figure}
\begin{4Dconsistency}
Given points $x_0$, $x_1$, $x_2$, $x_3$ and $x_4$
in general position in $\PP^M$, $M\geq 4$, choose generic points 
$x_{ij}\in\langle x_0, x_i, x_j \rangle$, $1\leq i < j \leq 4$,
on the corresponding planes. Using
the planarity condition construct the points 
$x_{ijk}\in\langle x_0, x_i, x_j , x_k\rangle$, $1\leq i < j < k \leq 4$ -- the
remaining vertices of the four (combinatorial) cubes. Then there are four
different
ways to construct the point $x_{1234}$, which is the last vertex of the 
(combinatorial) hypercube. However all of them give the same result due to the
fact that the point $x_{1234}$ is the unique
intersection point of the four three dimensional subspaces
$\langle x_{1}, x_{12}, x_{13}, x_{14} \rangle$,
$\langle x_{2}, x_{12}, x_{23}, x_{24} \rangle$,
$\langle x_{3}, x_{13}, x_{23}, x_{34} \rangle$,
and
$\langle x_{4}, x_{14}, x_{24}, x_{34} \rangle$ of the four dimensional subspace
$\langle x_{0}, x_{1}, x_{2}, x_{3} , x_{4} \rangle$.
\end{4Dconsistency}

\subsection{The backward data and the connection factors}
\label{sec:tQL}
The backward tangent vectors $\tbX_i$, the backward Lam\'e coefficients
$\tH_i$, $i=1,\dots,N$ and the backward rotation coefficients $\tQ_{ij}$ 
are defined with the help of the backward
shifts $T_i^{-1}$. They are again chosen 
in such a way that the
$T_i^{-1}$ variation of $\tbX_j$ is proportional to $\tbX_i$ only:
\begin{equation} \label{eq:lin-bX}
\quad \D_i\tbX_j =  (T_i\tbX_i)\tQ_{ij},    \quad i\ne j \; .
\end{equation}
Then
\begin{equation} \label{eq:b-H-X}
\quad \D_i\bx = (T_i\tbX_i ) \tH_i ,
\end{equation}
and
\begin{equation} \label{eq:lin-bH}
\D_j\tH_i =  (T_j\tQ_{ij})\tH_{j},   \quad i\ne j .
\end{equation}
The new functions $\tQ_{ij}$ satisfy the backward Darboux (MQL) 
equations
\begin{equation} \label{eq:MQL-tQ}
\D_k \tQ_{ij} = (T_k \tQ_{ik}) \tQ_{kj} ,    \qquad i, j, k \quad \text{distinct}.
\end{equation}
\begin{Rem}
Notice that, opposite to the commutative case \cite{KoSchief2,DS-sym},
the backward Darboux equations are not the same like the forward Darboux
equations~\eqref{eq:MQL-Q}.
\end{Rem}

The connection factors $\rho_i:\ZZ^N\to \DD$ are the 
proportionality coefficients between
$\bX_i$ and $T_i\tbX_i$ (both vectors are proportional to $\D_i\bx$):
\begin{equation} \label{eq:def-rho}
\bX_i =  -( T_i\tbX_i) \rho_i , \qquad 
T_iH_i = - \rho_i^{-1}\tH_i, \quad 
i=1,\dots ,N . 
\end{equation}
Going around an elementary quadrilateral it is not difficult to show that
\begin{equation} \label{eq:Q-Qt}
\rho_j T_j Q_{ij} = ( T_i\tQ_{ji}) \rho_i  ,
\end{equation}
and
\begin{equation} \label{eq:rho-constr}
T_j\rho_i = \rho_i ( 1 - (T_iQ_{ji})(T_jQ_{ij})) =
( 1 - (T_j\tQ_{ij})(T_i\tQ_{ji}) ) \rho_i, \quad i\ne j \; .
\end{equation}
\begin{Rem}
In the commutative case there exists yet another potential
(the $\tau$-function of the quadrilateral lattice) such that
\begin{equation*} \label{eq:tau}
\rho_i = \frac{T_i \tau}{\tau} ,
\end{equation*}
which is an immediate consequence of 
\begin{equation} \label{eq:rho-rho}
\frac{T_i \rho_j}{\rho_j} = \frac{T_j \rho_i}{\rho_i}.
\end{equation}
The last equation does not hold in the noncommutative case 
because, in general $(T_iQ_{ji})(T_jQ_{ij}) \neq (T_jQ_{ij})(T_iQ_{ji})$.
\end{Rem}

\section{Transformations of the quadrilateral lattice} \label{sec:FT}
Due to its vectorial character, the theory of transformations of quadrilateral 
lattices transfers to the noncommutative case almost without changes.
Therefore mostly we just state the relevant formulas (the proofs are by direct
verification along lines given in \cite{MDS,TQL,MM}).

Given the solution $\boldsymbol{Y}_i:\mathbb{Z}^N\to\DD^K$, 
of the linear system \eqref{eq:lin-X}, and given the solution 
$\boldsymbol{Y}^*_i:\mathbb{Z}^N\to(\DD^K)^*$, of the linear system 
\eqref{eq:lin-H}; we recall that elements of $\DD^K$ we represent by column
vectors, and elements of its dual $(\DD^K)^*$ as row vectors. These allow to
construct the linear operator valued potential 
$\boldsymbol{\Omega}(\boldsymbol{Y},\boldsymbol{Y}^*):
\mathbb{Z}^N\to M^K_K(\DD)$,
defined by 
\begin{equation} \label{eq:Omega-Y-Y}
\Delta_i \boldsymbol{\Omega}(\boldsymbol{Y},\boldsymbol{Y}^*) = 
\boldsymbol{Y}_i \otimes T_i\boldsymbol{Y}^*_i, 
\qquad i = 1,\dots , N;
\end{equation} 
similarly, one defines 
$\boldsymbol{\Omega}(\boldsymbol{X},\boldsymbol{Y}^*):
\mathbb{Z}^N\to M^M_K(\DD)$ and 
$\boldsymbol{\Omega}(\boldsymbol{Y},H):
\mathbb{Z}^N\to \DD^K$ by
\begin{align} \label{eq:Omega-X-Y}
\Delta_i \boldsymbol{\Omega}(\boldsymbol{X},\boldsymbol{Y}^*) & = 
\boldsymbol{X}_i \otimes T_i \boldsymbol{Y}^*_i, \\
\Delta_i\boldsymbol{\Omega}(\boldsymbol{Y},H)  & =
\boldsymbol{Y}_i \otimes T_i H_i. 
\end{align} 
We remark that because we multiply vectors from the right then covectors are
multiplied from the left. This makes the tensor products above well defined.
\begin{Prop} \label{prop:FT}
If $\boldsymbol{\Omega}(\boldsymbol{Y},\boldsymbol{Y}^*)$ is invertible then
the vector function $\bx^\prime:\ZZ^N\to\DD^M$ given by
\begin{equation} \label{eq:fund-vect}
\bx^\prime  = \bx - 
\boldsymbol{\Omega}(\boldsymbol{X},\boldsymbol{Y}^*)
\boldsymbol{\Omega}(\boldsymbol{Y},\boldsymbol{Y}^*)^{-1}
\boldsymbol{\Omega}(\boldsymbol{Y},H),
\end{equation}
represent a quadrilateral lattice (the vectorial fundamental transform of $x$),
whose Lam\'e coefficients $H_i^\prime$, normalized tangent vectors
$\bX_i^\prime$ and rotation coefficients $Q_{ij}^\prime$ are given by
\begin{align} 
\label{eq:fund-vect-H}
H_i^\prime  &= H_i - 
\boldsymbol{Y}^*_i
\boldsymbol{\Omega}(\boldsymbol{Y},\boldsymbol{Y}^*)^{-1}
\boldsymbol{\Omega}(\boldsymbol{Y},H),\\
\label{eq:fund-vect-X}
\bX^\prime_i  & = \bX_i - 
\boldsymbol{\Omega}(\boldsymbol{X},\boldsymbol{Y}^*)
\boldsymbol{\Omega}(\boldsymbol{Y},\boldsymbol{Y}^*)^{-1}
\boldsymbol{Y}_i,\\
\label{eq:fund-vect-Q}
Q_{ij}^\prime  & = Q_{ij} - 
\boldsymbol{Y}^*_j
\boldsymbol{\Omega}(\boldsymbol{Y},\boldsymbol{Y}^*)^{-1}
\boldsymbol{Y}_i.
\end{align}
Moreover, the backward data and the connection coefficients transform according
to
\begin{align} 
\label{eq:fund-vect-tH}
\tH_i^\prime  &= \tH_i +\rho_i 
\boldsymbol{Y}^*_i
\boldsymbol{\Omega}(\boldsymbol{Y},\boldsymbol{Y}^*)^{-1}
\boldsymbol{\Omega}(\boldsymbol{Y},H),\\
\label{eq:fund-vect-tX}
\tbX^\prime_i  & = \tbX_i + 
\boldsymbol{\Omega}(\boldsymbol{X},\boldsymbol{Y}^*)
\boldsymbol{\Omega}(\boldsymbol{Y},\boldsymbol{Y}^*)^{-1}
\boldsymbol{Y}_i\rho_i^{-1},\\
\label{eq:fund-vect-tQ}
\tQ_{ij}^\prime  & = \tQ_{ij} - 
\rho_i \boldsymbol{Y}^*_i
\boldsymbol{\Omega}(\boldsymbol{Y},\boldsymbol{Y}^*)^{-1}
\boldsymbol{Y}_j \rho_j^{-1},\\
\label{eq:fund-vect-rho}
\rho^\prime_i & = \rho_i(1 + T_i \boldsymbol{Y}^*_i
\boldsymbol{\Omega}(\boldsymbol{Y},\boldsymbol{Y}^*)^{-1}
\boldsymbol{Y}_i).
\end{align}
\end{Prop}
\begin{Rem}
We would like to mention that the above formulas can be put into a form using
the so called quasideterminants \cite{GGRW-quasideterminants}, like it was done, 
for example, in \cite{LiNimmo} for a non-Abelian Toda lattice.
\end{Rem}
\begin{Rem}
As it was shown in \cite{TQL} for the commutative case,
other Darboux-type transformations of the
quadrilateral lattice, like the Laplace, Combescure,
L\'{e}vy, adjoint L\'{e}vy or the radial transformations, can  
be obtained as reductions of the fundamental transformation. There are no
obstructions which would prevent the geometric reasoning applied in \cite{TQL}
to transfer such a statement to the noncommutative case. 
\end{Rem}

The vectorial
fundamental transformation can be considered as superposition of
$\dim\mathbb{V}$ (scalar) fundamental transformations; on intermediate stages
the rest of the transformation data should be suitably transformed as well.
Such a description contains already the principle of permutability of such
transformations, which follows from the following observation~\cite{TQL}.
\begin{Prop}
Assume the following splitting of the data of the vectorial fundamental
transformation
\begin{equation}
\boldsymbol{Y}_i = \left( \begin{array}{c} 
\boldsymbol{Y}_i^a \\ \boldsymbol{Y}_i^b \end{array} \right),\qquad
\boldsymbol{Y}_i^* = \left( \begin{array}{cc} 
\boldsymbol{Y}_{ai}^{*}\; & \boldsymbol{Y}_{b i}^{*} \end{array} \right),
\end{equation}
associated with the partition $\DD^K = \DD^{K_a} \oplus \DD^{K_b}$,
which implies the following splitting of the potentials
\begin{equation} \label{eq:split-fund-1}
\boldsymbol{\Omega}(\boldsymbol{Y},H) =  \left( \begin{array}{c} 
\boldsymbol{\Omega}(\boldsymbol{Y}^a,H) \\ 
\boldsymbol{\Omega}(\boldsymbol{Y}^b,H) \end{array} \right), \qquad
\boldsymbol{\Omega}(\boldsymbol{Y},\boldsymbol{Y}^*) = \left( \begin{array}{cc}
\boldsymbol{\Omega}(\boldsymbol{Y}^a,\boldsymbol{Y}_a^*) &
\boldsymbol{\Omega}(\boldsymbol{Y}^a,\boldsymbol{Y}_b^*) \\
\boldsymbol{\Omega}(\boldsymbol{Y}^b,\boldsymbol{Y}_a^*) &
\boldsymbol{\Omega}(\boldsymbol{Y}^b,\boldsymbol{Y}_b^*)\end{array} \right), 
\end{equation}
\begin{equation} \label{eq:split-fund-2}
\boldsymbol{\Omega}(\boldsymbol{X},\boldsymbol{Y}^*) = \left( \begin{array}{cc} 
\boldsymbol{\Omega}(\boldsymbol{X},\boldsymbol{Y}_a^*)\; , &
\boldsymbol{\Omega}(\boldsymbol{X},\boldsymbol{Y}_b^*)\end{array} \right).
\end{equation} 
Then the vectorial fundamental transformation is equivalent to the following
superposition of vectorial fundamental transformations:\\
1) Transformation $\bx\to\bx^{\{a\}}$ with the data 
$\boldsymbol{Y}_i^a$, $\boldsymbol{Y}_{ai}^*$ and the corresponding
potentials
$\boldsymbol{\Omega}(\boldsymbol{Y}^a,H)$, 
$\boldsymbol{\Omega}(\boldsymbol{Y}^a,\boldsymbol{Y}_a^*)$, 
$\boldsymbol{\Omega}(\boldsymbol{X},\boldsymbol{Y}_a^*)$
\begin{align}
\label{eq:fund-vect-a}
\bx^{\{a\}}  & = \bx - 
\boldsymbol{\Omega}(\boldsymbol{X},\boldsymbol{Y}^*_a)
\boldsymbol{\Omega}(\boldsymbol{Y}^a,\boldsymbol{Y}^*_a)^{-1}
\boldsymbol{\Omega}(\boldsymbol{Y}^a,H),\\
\boldsymbol{X}_i^{\{a\}} & = \boldsymbol{X}_i -
\boldsymbol{\Omega}(\boldsymbol{X},\boldsymbol{Y}^*_a)
\boldsymbol{\Omega}(\boldsymbol{Y}^a,\boldsymbol{Y}^*_a)^{-1}
\boldsymbol{Y}^a_i,
\\
H_i^{\{a\}} & = H_i - \boldsymbol{Y}^*_{i a}
\boldsymbol{\Omega}(\boldsymbol{Y}^a,\boldsymbol{Y}^*_a)^{-1}
\boldsymbol{\Omega}(\boldsymbol{Y}^a,H).
\end{align}
2) Application on the result the vectorial fundamental transformation with the
transformed data
\begin{align}
{\boldsymbol{Y}}_i^{b\{a\}} & = \boldsymbol{Y}_i^b -
\boldsymbol{\Omega}(\boldsymbol{Y}^b,\boldsymbol{Y}^*_a)
\boldsymbol{\Omega}(\boldsymbol{Y}^a,\boldsymbol{Y}^*_a)^{-1}
\boldsymbol{Y}^a_i,
\\
{\boldsymbol{Y}}_{i b}^{*\{a\}} & = \boldsymbol{Y}_{i b}^* - 
\boldsymbol{Y}^*_{i a}
\boldsymbol{\Omega}(\boldsymbol{Y}^a,\boldsymbol{Y}^*_a)^{-1}
\boldsymbol{\Omega}(\boldsymbol{Y}^a, \boldsymbol{Y}_{b}^*),
\end{align}
and potentials
\begin{align} 
{\boldsymbol{\Omega}}(\boldsymbol{Y}^b,H)^{\{a\}} & =
\boldsymbol{\Omega}(\boldsymbol{Y}^b,H) - 
\boldsymbol{\Omega}(\boldsymbol{Y}^b,\boldsymbol{Y}^*_a)
\boldsymbol{\Omega}(\boldsymbol{Y}^a,\boldsymbol{Y}^*_a)^{-1}
\boldsymbol{\Omega}(\boldsymbol{Y}^a,H)=
\boldsymbol{\Omega}({\boldsymbol{Y}}^{b\{a\}},H^{\{a\}}),
\\
{\boldsymbol{\Omega}}(\boldsymbol{Y}^b,\boldsymbol{Y}^*_b)^{\{a\}} & =
\boldsymbol{\Omega}(\boldsymbol{Y}^b,\boldsymbol{Y}^*_b) - 
\boldsymbol{\Omega}(\boldsymbol{Y}^b,\boldsymbol{Y}^*_a)
\boldsymbol{\Omega}(\boldsymbol{Y}^a,\boldsymbol{Y}^*_a)^{-1}
\boldsymbol{\Omega}(\boldsymbol{Y}^a,\boldsymbol{Y}^*_b)=
\boldsymbol{\Omega}({\boldsymbol{Y}}^{b\{a\}},{\boldsymbol{Y}}_b^{*\{a\}}),
\\
{\boldsymbol{\Omega}}(\boldsymbol{X},\boldsymbol{Y}^*_b)^{\{a\}}  & =
\boldsymbol{\Omega}(\boldsymbol{X},\boldsymbol{Y}^*_b) - 
\boldsymbol{\Omega}(\boldsymbol{X},\boldsymbol{Y}^*_a)
\boldsymbol{\Omega}(\boldsymbol{Y}^a,\boldsymbol{Y}^*_a)^{-1}
\boldsymbol{\Omega}(\boldsymbol{Y}^a,\boldsymbol{Y}^*_b)=
\boldsymbol{\Omega}({\boldsymbol{X}}^{\{a\}},{\boldsymbol{Y}}_b^{*\{a\}}),
\label{eq:fund-vect-potentials-slit}
\end{align}
i.e.,
\begin{equation} \label{eq:fund-vect-a-b}
\bx^\prime = \bx^{\{a,b\}}  = \bx^{\{a\}} - 
{\boldsymbol{\Omega}}(\boldsymbol{X},\boldsymbol{Y}^*_b)^{\{a\}}
[\boldsymbol{\Omega}(\boldsymbol{Y}^b,\boldsymbol{Y}^*_b)^{\{a\}}]^{-1}
{\boldsymbol{\Omega}}(\boldsymbol{Y}^b,H)^{\{a\}}.
\end{equation}
\end{Prop}
\begin{proof} The transformation rules for the intermediate data and potentials
are consequence of proposition~\ref{prop:FT}.
Denote 
\begin{equation*}
\boldsymbol{\Omega}(\boldsymbol{Y},\boldsymbol{Y}^*) = 
\boldsymbol{\Omega}=
\left( \begin{array}{cc}
\boldsymbol{\Omega}^a_a &
\boldsymbol{\Omega}^a_b \\
\boldsymbol{\Omega}^b_a &
\boldsymbol{\Omega}^b_b\end{array} \right), \qquad
\boldsymbol{\Omega}(\boldsymbol{Y},H) =  \left( \begin{array}{c} 
\boldsymbol{\Omega}^a \\ 
\boldsymbol{\Omega}^b \end{array} \right), \qquad
\boldsymbol{\Omega}(\boldsymbol{X},\boldsymbol{Y}^*) = \left( \begin{array}{cc} 
\boldsymbol{\Omega}_a, &
\boldsymbol{\Omega}_b\end{array} \right),
\end{equation*}
and notice that
\begin{equation}
\boldsymbol{\Omega}=
\left( \begin{array}{cc}
{1}_a & 0 \\
\boldsymbol{\Omega}^b_a (\boldsymbol{\Omega}^a_a)^{-1} &
1_b\end{array} \right)
\left( \begin{array}{cc}
\boldsymbol{\Omega}^a_a &
\boldsymbol{\Omega}^a_b \\
0 &
(\boldsymbol{\Omega}^b_b)^{\{ a\}} \end{array} \right),
\end{equation}
which gives
\begin{equation}
\boldsymbol{\Omega}^{-1}=
\left( \begin{array}{cc}
(\boldsymbol{\Omega}^a_a)^{-1} &
- (\boldsymbol{\Omega}^a_a)^{-1} \boldsymbol{\Omega}^a_b 
((\boldsymbol{\Omega}^b_b)^{\{ a\}})^{-1}\\
0 &
((\boldsymbol{\Omega}^b_b)^{\{ a\}})^{-1} \end{array} \right)
\left( \begin{array}{cc}
{1}_a & 0 \\
-\boldsymbol{\Omega}^b_a (\boldsymbol{\Omega}^a_a)^{-1} &
1_b\end{array} \right).
\end{equation}
Inserting such $\boldsymbol{\Omega}^{-1}$ into formula \eqref{eq:fund-vect} we
obtain
\begin{equation}
\bx^\prime = \bx - \left( \begin{array}{cc} 
\boldsymbol{\Omega}_a (\boldsymbol{\Omega}^a_a)^{-1} ,&
(\boldsymbol{\Omega}_b)^{\{ a\}} ((\boldsymbol{\Omega}^b_b)^{\{ a\}})^{-1} 
\end{array} \right)
\left( \begin{array}{c} 
\boldsymbol{\Omega}^a \\ 
(\boldsymbol{\Omega}^b)^{\{ a\}} \end{array} \right),
\end{equation}
thus equation \eqref{eq:fund-vect-a-b}.
\end{proof}
\begin{Rem}
The same result $\bx^\prime = \bx^{\{a,b\}}={\bx}^{\{b,a\}}$
is obtained exchanging the order of transformations, exchanging also the indices
$a$ and $b$ in formulas 
\eqref{eq:fund-vect-a}-\eqref{eq:fund-vect-a-b}.
\end{Rem}
\begin{Rem}
The scalar, i.e. $K=1$, fundamental transformation preserves in the
noncommutative case its geometric meaning as a transformation between two
quadrilateral lattices such that $x$, $x^\prime$ $T_i x$, $T_i x^\prime$ are
coplanar. Therefore also in the noncommutative case
the fundamental transformation can be considered as a construction of a new
level (in the new dimension direction) of the quadrilateral lattice. In
particular, in the case $K=2$, $K_a = K_b = 1$, any point $x$ of the lattice and
its transforms $x^{\{ a\} }$,   $x^{\{ b\} }$ and  $x^{\{ a,b\} }$ are coplanar. 
\end{Rem}

\section{The B-(Moutard) quadrilateral lattice}
\label{sec:BC}
We will concentrate below on the B-(Moutard)
quadrilateral lattice which provides geometric
interpretation of the discrete BKP equations. We will study implications of 
the corresponding additional (apart
from the Geometric Integrability Scheme) 
incidence geometric structures, which assure integrability of the above mentioned
reduction, on the possibility of deriving their noncommutative versions. The
main result of this Section is that the multidimensional consistency of the
reduction holds if and only if the division ring is commutative. 

In the geometric
considerations below we assume generality of configurations, i.e., only those
explicitly stated (and their consequences) hold. In particular, the subspace
\begin{equation*}
\langle x_{0}, x_{1}, x_{2}, x_{3}, x_{4} \rangle 
= \langle x_{1234}, x_{123}, x_{124}, x_{134}, x_{234} \rangle 
\end{equation*}
of the hypercube in the
Four Dimensional Consistency of the Geometric Integrability Scheme has dimension
four. 

\subsection{The B-quadrilateral lattice}
The B-quadrilateral lattice was defined geometrically in \cite{BQL} in the
commutative case (we consider for a moment the projective space over a (commutative)
field $\mathbb{F}$) as follows.
\begin{Def} \label{def:BQL}
A quadrilateral lattice $x:\ZZ^N\to\PP^M(\mathbb{F})$ is called the 
\emph{B-quadrilateral
lattice} if for any triple of different indices $i,j,k$
the points $x$, $T_i T_j x$, $T_i T_k x$ and $T_j T_k x$ are coplanar.
\end{Def}
In \cite{BQL} it was also
shown that the homogeneous coordinates $\bx:\ZZ^N\to \mathbb{F}^{M+1}_*$ satisfy
(in appropriate gauge) the 
system of discrete Moutard equations \cite{DJM-D5,NiSchief} 
\begin{equation} \label{eq:BKP-lin}
T_i T_j \bx - \bx = \frac{(T_i \tau)T_j\tau}{\tau T_i T_j\tau}
(T_i \bx - T_i \bx) , \quad 1\leq i< j\leq N,
\end{equation}
where the $\tau$-function above is the square root of the $\tau$-function of the
quadrilateral lattice mentioned in the last remark of section~\ref{sec:tQL}.
The compatibility condition of the linear system \eqref{eq:BKP-lin} is 
Miwa's discrete BKP system \cite{Miwa}
\begin{equation} \label{eq:BKP-nlin}
\tau\, T_i T_j T_k \tau= (T_i T_j \tau)T_k\tau - (T_i T_k\tau)T_j\tau + 
(T_j T_k \tau)T_i\tau, \quad 1\leq i< j < k \leq N.
\end{equation}

Because the B-reduction condition is imposed on the elementary hexahedra
level, to show integrability of the B-quadrilateral lattice
it is important to check its four dimensional compatibility with the
Geometric Integrability Scheme. The four dimensional consistence of the 
BQL-constraint was proved algebraically
in \cite{BQL} in the commutative case.
We will show geometrically that, in contrary to the
quadrilateral lattice case, one cannot obtain directly the noncommutative
integrable B-quadrilateral lattice.
\begin{Th}
Multidimensional consistency of the B-quadrilateral lattice constraint holds if
and only if the division ring $\mathbb{D}$ is commutative.
\end{Th}
\begin{proof}
It is an immediate consequence of two Lemmas below.
\end{proof}
\begin{Lem} \label{lem:BQL-1}
The B-constraint is multidimensionally consistent if and only if 
for any triple of different indices $i,j,k$
the points $ T_ix$, $ T_j x$, $T_k x$ and $T_i T_j T_k x$ are coplanar
as well.
\end{Lem}
\begin{Lem} \label{lem:BQL-comm}
Under hypotheses of the Geometric Integrability Scheme, 
assume that $x_0$, $x_{12}$, $x_{13}$ and $x_{23}$ are coplanar.
Then the following is true: $\DD$ is \emph{commutative} 
(hence a field) if and only if the points $x_1$, $x_2$, $x_3$ and
$x_{123}$ are coplanar as well (see Figure~\ref{fig:BQLconditon}).
\end{Lem}
\begin{figure}
\begin{center}
\includegraphics{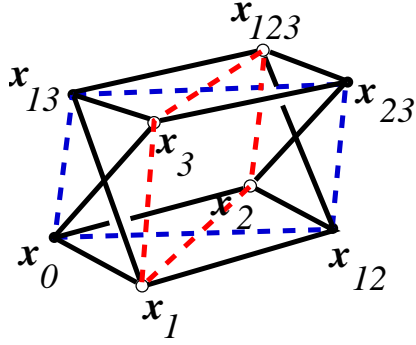}
\end{center}
\caption{Elementary hexahedron of the B-quadrilateral lattice}
\label{fig:BQLconditon}
\end{figure}
\begin{proof}[Proof of Lemma~\ref{lem:BQL-1}]
Consider a hypercube with planar faces as in Four Dimensional Consistency of the
Geometric Integrability Scheme. It consists with four ``initial
hexahedra" shearing vertex $x_0$, and the four ``final hexahedra" 
shearing vertex $x_{1234}$. 

To demonstrate the first implication consider three ``final hexahedra"
containing the vertex $x_{123}$. Because by the B-reduction condition 
\begin{equation*}
x_1 \in \langle x_{123}, x_{134}, x_{124} \rangle, \qquad
x_2 \in \langle x_{123}, x_{124}, x_{234} \rangle, \qquad
x_3 \in \langle x_{123}, x_{134}, x_{234} \rangle, 
\end{equation*}
then the three planes above (and therefore the points $x_1$, $x_2$ and $x_3$)
are contained in the three dimensional subspace
$\langle x_{123}, x_{124}, x_{134},  x_{234}\rangle$. Notice that this subspace
contains also $x_4$ (as a point of the plane 
$\langle x_{124}, x_{134},  x_{234}\rangle$).
Another three dimensional subspace 
$\langle x_{123}, x_{12}, x_{13},  x_{23}\rangle$
is the subspace of the initial hexahedron containing $x_{123}$. 
By construction (according to the Geometric Integrability Scheme)
it contains also the points $x_{1}$, $x_{2}$ and $x_{3}$. Both
subspaces are different (one contains $x_4$ and the other does not), 
and belong to the four dimensional subspace
$\langle x_{0}, x_{1}, x_{2}, x_{3}, x_{4} \rangle $ of the hypercube.
Therefore their intersect in a plane. Therefor the points  
$x_1$, $x_2$, $x_3$ and $x_{123}$ are coplanar, i.e., the implication holds for
one of the ``initial hexahedra"; the statement for three
others can be shown analogously.

To show the backward implication we apply similar arguments, but for a ``final
hexahedron" of the hypercube --- this time let us concentrate on that containing
$x_1$. 
Notice that, by the assumption, the three planes 
\begin{equation*}
\langle x_{1}, x_{2}, x_{3} \rangle, \qquad
\langle x_{1}, x_{2}, x_{4} \rangle, \qquad
\langle x_{1}, x_{3}, x_{4} \rangle.
\end{equation*} 
contain, respectively, $x_{123}$, $x_{124}$ and $x_{134}$.
They belong therefore to the subspace 
$\langle x_{1}, x_{2}, x_{3}, x_{4} \rangle$ of dimension three, which contains
also the point $x_{234}$ (as a point of the plane 
$\langle x_{2}, x_{3},  x_{4}\rangle$).
Another three dimensional subspace 
$\langle x_{1}, x_{12}, x_{13}, x_{14}\rangle$, of the final hexahedron we are
considering, by construction (according to the Geometric Integrability Scheme)
also contains the points $x_{123}$, $x_{124}$ and $x_{134}$. Notice that this
subspace cannot contain $x_{234}$, because
it would contain then all the vertices of the hypercube). Both
subspaces are different, and belong to the four dimensional subspace
\begin{equation*}
\langle x_{0}, x_{1}, x_{2}, x_{3}, x_{4} \rangle 
=\langle x_{1234}, x_{123}, x_{124}, x_{134}, x_{234} \rangle ,
\end{equation*}
then they both intersect in a plane. This plane contains the points  
$x_1$, $x_{123}$, $x_{124}$ and $x_{134}$, which shows that the hexahedron under
investigation  satisfies the B-reduction condition; the statement for three
others can be shown analogously.
\end{proof}

The geometric proof of Lemma~\ref{lem:BQL-comm} can be obtained by application: 
(i) its equivalence with certain theorem concerning the so called quadrangular
set of points \cite{BQL}, an (ii) equivalence of that theorem with validity
of the Pappus' configuration \cite{Coxeter-PG}. Below we give a direct 
algebraic proof. 
\begin{proof}[Algebraic proof of Lemma~\ref{lem:BQL-comm}] 
In what follows, by $\bx\in\mathbb{D}^{M+1}_*$ we denote the homogeneous 
coordinates of a point
$x\in\mathbb{P}^M(\mathbb{D})$; recall that we deal with right vector spaces.

The coplanarity of the four points $x_0$, $x_1$, $x_2$ and $x_{12}$ can be
algebraically expressed as the linear relation 
\begin{equation*}
\bx_{0} \alpha+ \bx_{1}\beta + \bx_{2} \gamma+ \bx_{12}\delta = 0, 
\end{equation*}
where, by the generality assumption (no three of the points are collinear), 
all the coefficients do not vanish. Suitably rescaling the homogeneous
coordinates of the points we can transfer above equation 
to the form 
\begin{equation} \label{eq:BQL-gauge-initial-12}
\bx_{12} = \bx_0 + \bx_{1} + \bx_{2} .
\end{equation}
In the equation expressing coplanarity of the points $x_0$, $x_1$, $x_3$ and 
$x_{13}$ we can again rescale the homogeneous
coordinates of $x_3$ and $x_{13}$ to get
\begin{equation} \label{eq:BQL-gauge-initial-13}
\bx_{13} = \bx_0 + (\bx_{1} + \bx_{3}) a.
\end{equation}
However, the coplanarity of 
$x_0$, $x_2$, $x_3$ and $x_{23}$ can be
expressed, by plaing with the gauge of $\bx_{23}$, at most as
\begin{equation} \label{eq:gauge-23-bad}
\bx_{23} = \bx_0 + \bx_{2} b + \bx_{3}c. 
\end{equation}
Then the additional
condition of coplanarity of $x_0$, $x_{12}$, $x_{13}$ and $x_{23}$,
which is equivalent to existence of $\lambda,\mu,\nu\in\mathbb{D}$ such that the
expression
\begin{equation*}
\bx_{12}\lambda + \bx_{13} \mu + \bx_{23} \nu 
\end{equation*}
is proportional to $\bx_0$, gives $c=-b$.

The homogeneous coordinates of the point $x_{123}$ are given by
\begin{equation*}
\bx_{123} = \bx_{1}A + \bx_{12}B +\bx_{13}C =
\bx_{2}\tilde{A} + \bx_{23}\tilde{B} +\bx_{12}\tilde{C} = 
\bx_{3}A^\prime+ \bx_{13}B^\prime +\bx_{23}C^\prime,
\end{equation*}
where the nine coefficients $A$,\dots ,$C^\prime$ (to be determined) are given
up to a common factor. Using equations 
\eqref{eq:BQL-gauge-initial-12}, \eqref{eq:BQL-gauge-initial-13}
and  \eqref{eq:gauge-23-bad} with $c=-b$ we obtain 
decomposition of $\bx_{123}$ in terms of the basis vectors $\bx_0, \bx_1, \bx_2,
\bx_3$
\begin{eqnarray*}
\bx_{123} &= & \bx_0 (B+C) + \bx_1 (A + B + aC) + \bx_2 B + \bx_3 aC , \\
&= & \bx_0 (\tilde{B} + \tilde{C}) + \bx_1 \tilde{C} + \bx_2 (\tilde{A} +
b\tilde{B} +\tilde{C}) - \bx_3 b \tilde{B} , \\
&= & \bx_0 (B^\prime + C^\prime) + \bx_1 a B^\prime + \bx_2 b C^\prime  + \bx_3
(A^\prime + aB^\prime - b C^\prime). 
\end{eqnarray*}
In consequence we obtain eight equations 
\begin{gather*}
B+C = \tilde{B} + \tilde{C} = B^\prime + C^\prime, \qquad
A + B + aC = \tilde{C} = a B^\prime, \\
B = \tilde{A} + b\tilde{B} +\tilde{C} = b C^\prime, \qquad
aC = - b \tilde{B} = A^\prime + aB^\prime - b C^\prime,
\end{gather*}
which allow to find the coefficients $A$,\dots ,$C^\prime$. 

The additional requirement $x_{123}\in \langle x_1, x_2 , x_3 \rangle$
algebraically means that
the coefficients in front of $\bx_0$ in the above decompositions vanish.
Neglecting three of the above equations which simply express $A$,
$\tilde{A}$ and $A^\prime$ in terms of six other coefficients
$B$,\dots , $C^\prime$, we obtain the system
\begin{gather*}
B+C = \tilde{B} + \tilde{C} = B^\prime + C^\prime = 0 , \\
\tilde{C} = a B^\prime, \qquad B = b C^\prime, \qquad
aC = - b \tilde{B},
\end{gather*}
which allows for nontrivial solution if and only if $ab = ba$.
\end{proof}

\section{Conclusions and discussion}
Motivated by validity of the Geometric Integrability Scheme in projective spaces
over division rings we investigated basic properties of the quadrilateral 
lattices in such spaces and the corresponding
version of the discrete Darboux equations. In particular, we showed that basic
ingredients of the vectorial fundamental transformation of quadrilateral
lattices transfer to such a setting almost without changes (one has to take care
of correct ordering only). We would like to mention that in the incidence 
geometry one considers also
more general spaces over rings \cite{Veldkamp}, which should provide geometric
interpretation for the matrix Darboux equations.

We also investigated possibility of obtaining the noncommutative version of the
B-(Moutard) quadrilateral lattices. It turns out that the
additional incidence geometry assumptions which imply integrability of such
lattices hold if and only if the division ring under consideration is
commutative (hence a field). The question remains open for more general 
geometries over rings.

\section*{Acknowledgements}
The author acknowledges numerous discussions with Jaros{\l}aw Kosiorek, Andrzej 
Matra\'{s} and Mark Pankov on foundations of geometry. The paper was supported 
by the Polish Ministry of Science and Higher Education research grant 
1~P03B~017~28. 

\bibliographystyle{amsplain}

\providecommand{\bysame}{\leavevmode\hbox to3em{\hrulefill}\thinspace}

\end{document}